\documentclass[10pt]{article}
\usepackage{graphicx,indentfirst}
\usepackage{amssymb,amsmath,amsthm,algorithm,algorithmic,enumerate}

\newtheorem{theorem}{Theorem}

\newtheorem{prop}[theorem]{Proposition}
\newtheorem{lemma}[theorem]{Lemma}
\newtheorem{claim}[theorem]{Claim}
\newcommand{\opt}{\mathit{OPT}}
\newcommand{\bigR}{\mathbb{R}}

\newcommand{\lpt}{\textsl{LPT* }}
\newcommand{\ep}[1]{\mathbb{E}[#1]}
\newcommand{\dx}{\ dx}
\newcommand{\dt}{\ dt}
\newcommand{\da}{\ d\alpha}
\newcommand{\qcmax}{Q||C_{max}}
\newcommand{\etal}{et$.$ al$.$ }
\title{Lower Bound for Envy-Free and Truthful Makespan Approximation on Related Machines
\footnote{This work was supported in part by NSF grants CCF-0728869 and CCF-1016778. Contact information: \{lkf,zhenghui\}@cs.dartmouth.edu. Dept of Computer Science, Dartmouth College, Hanover, NH 03755, USA.
 An extended abstract will appear in SAGT 2011.}
}
\author{Lisa Fleischer
 \and Zhenghui Wang
}
\date{July 14, 2011}
\begin{document}
\maketitle
\begin{abstract}
We study problems of scheduling jobs on related machines so as to minimize the makespan in the setting where machines are strategic agents. In this problem, each job $j$ has a length $l_{j}$ and each machine $i$ has a private speed $t_{i}$. The running time of job $j$ on machine $i$ is $t_{i}l_{j}$. We seek a mechanism that obtains speed bids of machines and then assign jobs and payments to machines so that the machines have incentive to report true speeds and the allocation and payments are also envy-free. We show that
\begin{enumerate}
\item A deterministic envy-free, truthful, individually rational, and anonymous mechanism cannot approximate the makespan strictly better than $2-1/m$, where $m$ is the number of machines. This result contrasts with prior work giving a deterministic PTAS for envy-free  anonymous assignment and a distinct deterministic PTAS for truthful anonymous mechanism.
\item For two machines of different speeds, the unique deterministic scalable allocation of any envy-free, truthful, individually rational, and anonymous mechanism is to allocate all jobs to the quickest machine. This allocation is the same as that of the VCG mechanism, yielding a 2-approximation to the minimum makespan.
\item
No payments can make any of the prior published monotone and locally efficient allocations that yield better than an $m$-approximation for $\qcmax$~\cite{aas, at,ck10, dddr, kovacs} a truthful, envy-free, individually rational, and anonymous mechanism.
\end{enumerate}
\end{abstract}
\newpage

\setlength{\parskip}{0.3\baselineskip}
\section{Introduction}
We study problems of scheduling jobs on related machines so as to minimize the makespan (i.e. $\qcmax$) in a strategic environment. Each job $j$ has a length $l_{j}$ and each machine $i$ has a \emph{private} speed $t_{i}$, which is only known by that machine. The speed $t_{i}$ is the time it takes machine $i$ to process one unit length of a job --- $t_{i}$ is the inverse of the usual sense of speed. The running time of job $j$ on machine $i$ is $t_{i}l_{j}$. A single job cannot be performed by more than one machine (indivisible), but multiple jobs can be assigned to a single machine. The \emph{workload} of a machine is the total length of jobs assigned to that machine and the \emph{cost} is the running time of its workload. The scheduler  would like to schedule jobs to complete in minimum time, but has to pay machines to run jobs. The \emph{utility} of a machine is the difference between the payment to the machine and its cost. The mechanism used by the scheduler asks the machines for their speeds and then determines an allocation of jobs to machines and payments to machines. Ideally, the mechanism should be fair and efficient. To accomplish this, the following features of mechanism are desirable.

\begin{description}
\item[Individually rational]
A mechanism is \emph{individually rational} (IR), if no agent gets negative utility when reporting his true private information, since a rational agent will refuse the allocation and payment if his utility is negative.  In order that each machine accepts its allocation and payment, the payment to a machine should exceed its cost of executing the jobs.
\item[Truthful]
A mechanism is \emph{truthful} or  \emph{incentive compatible} (IC), if each agent maximizes his utility by reporting his true private information.  Under truth-telling, it is easier for the designer to design and analyze mechanisms, since agents' dominant strategies are known by the designer. In a truthful mechanism, an agent does not need to compute the strategy maximizing his utility, since it is simpler to report his true information.
\item[Envy-free]
A mechanism is \emph{envy-free} (EF), if no agent can improve his utility by switching his allocation and payment with that of another. Envy-freeness is a strong concept of fairness~\cite{ds61,F67}: each agent is happiest with his allocation and payment.
\end{description}

Prior work on envy-free mechanisms for makespan approximation problems assumes that all machine speeds are public knowledge~\cite{cohen10,ms}. We assume that the speed of a machine is private information of that machine. This assumption makes it harder to achieve envy-freeness. Only if the mechanism is also truthful, can the mechanism designer ensure that the allocation is truly envy-free.

In this paper, we prove results about \emph{anonymous} mechanisms. A mechanism is anonymous, roughly speaking, if  when two agents switch their bids, their allocated jobs and payments  also switch. This means the allocation and payments depend only on the agents' bids, not on their names. Anonymous mechanisms are of interest in this problem for two reasons. On the one hand, to the best of our knowledge, all polynomial-time mechanisms for $\qcmax$ are anonymous \cite{aas, at,ck10, dddr,ms}.
On the other hand, in addition to envy-freeness, anonymity can be viewed as an additional characteristic of fairness~\cite{adl}.

We also study \emph{scalable} allocations. Scalability means that multiplying  the speeds by the same positive constant does not change the allocation. Intuitively, the allocation function should not depend on the ``units'' in which the speed are measured, and hence scalability is a natural notion. But allocations based on rounded speeds of machines are typically not scalable~\cite{aas,ck10,kovacs}.

The truthful mechanisms and envy-free mechanisms for $\qcmax$ are both well-understood.
There is a payment scheme to make an allocation truthful if and only if the allocation is \emph{monotone decreasing}~\cite{at}. For $\qcmax$, an allocation is monotone decreasing if no machine gets more workload by bidding a slower speed than its true speed. On the other hand, a mechanism for $\qcmax$ can be envy-free if and only if its allocation is \emph{locally efficient}~\cite{ms}. An allocation is locally efficient if a machine never gets less workload than a slower one.

The complexity of truthful mechanisms and, separately, envy-free mechanisms have been completely settled. $\qcmax$ is strongly NP-hard, so there is no FPTAS for this problem, assuming P $\neq$ NP. On the other hand, there is a deterministic monotone  PTAS~\cite{ck10} and a distinct deterministic locally efficient PTAS~\cite{ms}. This implies the existence of truthful mechanisms and distinct envy-free mechanisms that approximate the makespan arbitrarily closely. However, neither of these payment functions make the mechanisms both truthful and envy-free.

The VCG mechanism for $\qcmax$ is truthful, envy-free, individually rational, and anonymous~\cite{cffko}. However, since the VCG mechanism maximizes the social welfare (i.e. minimizing the total running time), it always allocates all jobs to the quickest machines, yielding  a $m$-approximation of makespan for $m$ machines in the worst case. So a question is whether there is a truthful, envy-free, individually rational and anonymous mechanism that approximates the makespan better than the VCG mechanism. Since there already exists many allocation functions that are both monotone and locally efficient, one natural step to answer this question could be checking whether some of these allocation functions admit truthful and envy-free payments.\par

\paragraph{Our Results.} We show that
\begin{enumerate}
\item A deterministic envy-free, truthful, individually rational, and anonymous mechanism cannot approximate the makespan strictly better than $2-1/m$, where $m$ is the number of machines. (Section 3). This result contrasts with prior results~\cite{ck10,ms} discussed above.
\item For two machines of different speeds, the unique deterministic scalable allocation of any envy-free, truthful, individually rational, and anonymous mechanism is to allocate all jobs to the quickest machine. (Section 5). This allocation is the same as that of the VCG mechanism, yielding a 2-approximation of makespan for this case.
\item
No payments can make any of the prior published monotone and locally efficient allocations that yield better than an $m$-approximation for $\qcmax$~\cite{aas, at,ck10, dddr, kovacs} a truthful, envy-free, individually rational, and anonymous mechanism.
\end{enumerate}
\paragraph{Related Work.}\
Hochbaum and Shmoys~\cite{hs} give a  PTAS for $\qcmax$. Andelman, Azar, and Sorani \cite{aas} give a 5-approximation deterministic truthful mechanism. Kov\'{a}cs improves the approximation ratio to 3~\cite{kovacs} and then to 2.8~\cite{kovacs09}. Randomization has been successfully applied to this problem.
Archer and Tardos~\cite{at} give a 3-approximate randomized mechanism, which is improved to 2 in~\cite{Archer}. Dhangwatnotai~\cite{dddr} \etal give a monotone randomized PTAS. All these randomized mechanisms are truthful-in-expectation. However, we can show that no payment function can form a truthful, envy-free, individually rational and anonymous mechanism with any allocation function of these mechanisms. We give a proof for a deterministic allocation~\cite{kovacs} in Section 5 and another one for a randomized allocation~\cite{at} in Appendix B.

When players have different finite valuation spaces, it is known that a monotone and locally efficient allocation function may not admit prices to form a simultaneously truthful and envy-free mechanism for allocating goods among players~\cite{cffko2}. In this paper, we consider mechanisms where all players have identical infinite valuation spaces.

Cohen \etal\cite{cffko} study the truthful and envy-free mechanisms on combinatorial auctions with additive valuations where agents have a upper capacity on the number of items they can receive. They seek truthful and envy-free mechanisms that maximize social welfare and show that VCG with Clarke Pivot payments is envy-free if agents' capacities are all equal. Their result can be interpreted in our setting by viewing that each  agent has the same capacity $n$ and the valuation of each agent is the reverse of its cost. So their result implies that the VCG mechanism for $\qcmax$ is truthful and envy-free; but the VCG mechanism does not give a good approximation guarantee for makespan.
\section{Preliminaries}
There are $m$ machines and $n$ jobs. Each agent will report a bid $b_{i}\in \bigR$ to the mechanism. Let $t$ denote the vector of true speeds and $b$ the vector of bids.

A mechanism consists of a pair of functions $(w,p)$. An \emph{allocation} $w$ maps a vector of bids to a vector of allocated workload, where $w_{i}(b)$ is the workload of agent $i$. For all bid vectors $b$, $w(b)$ must correspond to a valid job assignment. An allocation $w$ is called \emph{scalable} if $w_{i}(b)=w_{i}(c\cdot b)$ for all bid vectors $b$,  all $i\in\{1\ldots m\}$ and all scalars $c>0$. A \emph{payment} $p$ maps a vector of bids to a vector of payments, i.e. $p_{i}(b)$ is the payment to agent $i$.

The \emph{cost} machine $i$ incurs by the assigned jobs is $t_{i}w_{i}(b)$. Machine $i$'s private value $t_{i}$ measures its cost per unit work. Each machine $i$ attempts to maximize its \emph{utility}, $u_{i}(t_{i},b):=p_{i}(b)-t_{i}w_{i}(b)$.

The \emph{makespan} of allocation $w(b)$ is defined as $\max_{i}\  w_{i}(b)\cdot t_{i}$. A mechanism $(w,p)$ is $c$-\emph{approximate} if  for all bids $b$ and values $t$, the makespan of the allocation given by $w$ is within $c$ times the makespan of the optimal allocation, i.e.,
$$\textrm{max}{}_{i}\  w_{i}(b)\cdot t_{i}\leq c\cdot \opt(t),$$
where $\opt(t)$ is the minimum makespan for machines with speeds $t$ .

Vector $b$ is sometimes written as $(b_{i},b_{-i})$, where $b_{-i}$ is the vector of bids, not including agent $i$. A mechanism $(w,p)$ is \emph{truthful} or \emph{incentive compatible}, if each agent $i$ maximizes his utility by bidding his true value $t_{i}$, i.e., for all agent $i$, all possible $t_{i},b_{i}$ and $b_{-i}$, $$p_{i}(t_{i},b_{-i})-t_{i}w_{i}({t_{i},b_{-i}})\geq p_{i}(b_{i},b_{-i})-t_{i}w_{i}({b_{i},b_{-i}})$$

A mechanism $(w,p)$ is \emph{individually rational}, if agents who bid truthfully never incur a negative utility, i.e. $u_{i}(t_{i},(t_{i},b_{-i}))\geq 0$ for all agents $i$, true value $t_{i}$ and other agents' bids $b_{-i}$.

A mechanism $(w,p)$ is \emph{envy-free} if no agent wishes to switch his allocation and payment with another. For all $i,j \in\{1,\ldots, m\}$ and all bids $b$, $$p_{i}(b)-b_{i}w_{i}(b)\geq p_{j}(b)-b_{i}w_{j}(b).$$Notice that we use bids $b$ instead of the true speeds $t$ in this definition, because a mechanism can determine the envy-free allocation only based on the bids. However, a mechanism can ensure the outcome is envy-free, only if it is also truthful.

A mechanism $(w,p)$ is anonymous if for every bid vector $b=(b_{1},\ldots,b_{m})$, every $k$ such that $b_{k}$ is unique and every $l\neq k$,
$$w_{l}(\ldots,b_{k-1},b_{l},b_{k+1},\ldots,b_{l-1},b_{k},b_{l+1},\ldots)=w_{k}(b)$$ and
$$p_{l}(\ldots,b_{k-1},b_{l},b_{k+1},\ldots,b_{l-1},b_{k},b_{l+1},\ldots)=p_{k}(b).$$

The condition that $b_{k}$ is unique  is important, because in some case the mechanism may have to allocate jobs of different lengths to agents with the same bids. If  mechanism $(w,p)$ is anonymous and the bid of an agent is unique, the workload of that agent stays the same no matter how that agent is indexed. So we can write $w_{i}(b_{i},b_{-i})$ simply as $w(b_{i},b_{-i})$ for unique $b_{i}$ to represent the workload of agent $i$. Similarly, we can write $p_{i}(b_{i},b_{-i})$ simply as $p(b_{i},b_{-i})$ for unique $b_{i}$.

\subsubsection*{Characterization of truthful mechanisms}
\begin{lemma}[\cite{at}]
The allocation $w(b)$ admits a truthful payment scheme if and only if $w$ is monotone decreasing, i.e., $w_{i}(b_{i}',b_{-i})\leq w_{i}(b_{i},b_{-i})$ for all $i, b_{-i},b_{i}'\geq b_{i}$. In this case, the mechanism is truthful if and only if the payments satisfy
\begin{equation}\label{equic}
p_{i}(b_{i},b_{-i})=h_{i}(b_{-i})+b_{i}w_{i}(b_{i},b_{-i})-\int_{0}^{b_{i}}w_{i}(u,b_{-i})\ du,\qquad\forall i
\end{equation}
where the $h_{i}$s can be arbitrary functions.
\end{lemma}
By Lemma 1, the only flexibility in designing the truthful payments for allocation $w$ is to choose the terms $h_{i}(b_{-i})$. The utility of truth-telling agent $i$ is $h_{i}(b_{-i})-\int_{0}^{b_{i}}w_{i}(u,b_{-i})du$, because his cost is $t_{i}w_{i}(t_{i},b_{-i})$, which cancels out the second term in the payment formula. Thus, in order to make the mechanism individually rational, the term $h_{i}(b_{-i})$ should be at least $ \int_{0}^{b_{i}}w_{i}(u,b_{-i}) \ du$ for any $b_{i}$. Since $b_{i}$ can be arbitrarily large, $h_{i}$ should satisfy
\begin{equation}\label{equir}
h_{i}(b_{-i})\geq \int_{0}^{\infty}w_{i}(u,b_{-i}) \ du, \qquad \forall i,b_{-i}.
\end{equation}

\subsubsection*{Characterization of envy-free mechanisms}

An allocation function $w$ is \emph{envy-free implementable} if there exists a payment function $p$ such that the mechanism $M=(w,p)$ is envy-free.
An allocation function $w$ is \emph{locally efficient} if for all bids $b$, and all permutations $\pi$ of $\{1,\cdots,m\}$,
$$\sum_{i=1}^{m}b_{i}\cdot w_{i}(b)\leq \sum_{i=1}^{m}b_{i}\cdot w_{\pi(i)}(b).$$
\begin{lemma}[\cite{ms}]
Allocation $w$ is envy-free implementable if and only if $w$ is locally efficient.
\end{lemma}
The proof of sufficiency constructs a payment scheme that ensures the envy-freeness for any locally efficient allocation $w$. Specifically, assuming $b_{1}\geq b_{2}\geq\ldots\geq b_{m}$, the payments for related machines are the following:
$$
p_{i}(b)=\begin{cases}
b_{1}\cdot w_{1}(b) & \text{for } i=1\\
p_{i-1}(b)+b_{i}\cdot (w_{i}(b)-w_{i-1}(b)) & \text{for } i\in\{2,\ldots,m\} \\
\end{cases}
$$ These payments are not truthful payments, since $p_{1}(b)$ is clearly not in the form of~(\ref{equic}). But the set of envy-free payments is a convex polytope for fixed $w$, since payments satisfying linear constraints $\forall i,j~~p_{i}(b)-b_{i}w_{i}(b)\geq p_{j}(b)-b_{i}w_{j}(b)$ are envy-free. So there could be other payments that are both envy-free and truthful.

\section{Lower Bound on Anonymous Mechanisms}
In this section, we will prove an approximation lower bound for truthful, envy-free, individually rational, and anonymous mechanisms.
\begin{theorem}\label{thm1} Let $M=(w,p)$ be a deterministic, truthful, envy-free,
individually rational, and anonymous mechanism. Then $M$ is not $c$-approximate for $c<
2-\frac{1}{m}$. \end{theorem}
Since the only flexibility when designing payments in a
truthful mechanism is to choose the $h_{i}$s, we need to know what kind of
$h_{i}$s are required for envy-free anonymous mechanisms. The following two lemmas give
necessary conditions on $h_{i}$s. \begin{lemma}\label{lemma_1} If a mechanism $(w,p)$ is
both truthful and anonymous, then there is a function $h$ such that $h_{i}(v)=h(v)$ in
(\ref{equic}) for all bid vector $v\in\bigR^{m-1}_{+}$ and machines $i$. \begin{proof}
Let $\beta$ be a real number such that $\beta<\min_{j}v_{j}$. For all
$i\in\{1,\ldots,m-1\}$, define vector $b^{(i)}=(v_{1},\ldots,
v_{i-1},\beta,v_{i},\ldots,v_{m-1})$, $b'^{(i)}=(v_{1},\ldots,
v_{i},\beta,v_{i+1},\ldots,v_{m-1})$. Since $v=b^{(i)}_{-i}=b'^{(i)}_{-(i+1)}$ and $M$ is anonymous,
it must be that $p_{i}(b^{(i)})=p_{i+1}(b'^{(i)})$ and $w_{i}(b^{(i)})=w_{i+1}(b'^{(i)})$. Since $\alpha<v_{j}$
for any $0<\alpha<\beta, j\in\{1\ldots m-1\}$, we also have
$w_{i}(\alpha,v)=w_{i+1}(\alpha,v)$ by anonymity. Thus, for truthful payments, we have
$$h_{i}(v)~=~p_{i}(b^{(i)})-\beta
w_{i}(b^{(i)})+\int_{0}^{\beta}w_{i}(\alpha,v)\da~=~p_{i+1}(b'^{(i)})-\beta
w_{i+1}(b'^{(i)})+\int_{0}^{\beta}w_{i+1}(\alpha,v)\da~=~h_{i+1}(v).$$ \end{proof} \end{lemma}

\begin{lemma}\label{lemma_envy}
Let $L=\sum_{k}l_{k}$. If mechanism $M=(w,p)$ is truthful, envy-free, and anonymous, then
\begin{equation}\label{equenvy}
h(t_{-i})-h(t_{-j})\leq L\cdot t_{i}+(t_{j}-t_{i})w_i(t),
\end{equation}
for all $t\in \bigR_{+}^{m}$ and $i,j\in\{1,\ldots,m \}$.
\end{lemma}
\begin{proof}
If machine $j$ does not envy machine $i$, then $p_{j}(t)-t_{j}w_{j}(t) \geq p_{i}(t)-t_{j}w_{i}(t)$. Using~(\ref{equic}) to substitute in for $p_{i}$ and $p_{j}$, and Lemma~\ref{lemma_1}, this yields $$\left( h( t_{-j})+t_{j}w_j(t)-\int_{0}^{t_{j}}w(x,t_{-j})\ dx\right) - t_{j}w_j(t) \geq \left(h(t_{-i})+t_{i}w_i(t)-\int_{0}^{t_{i}}w(x,t_{-i})\ dx\right)-t_{j}w_i(t).$$
Rearranging terms gives
\begin{eqnarray*}
h(t_{-i})-h(t_{-j})&\leq& \int_{0}^{t_{i}}w(x,t_{-i}) \dx - \int_{0}^{t_{j}}w(x,t_{-j})dx+(t_{j}-t_{i})w_i(t)\\
&\leq&\int_{0}^{t_{i}}L \dx - \int_{0}^{t_{j}}0\dx+(t_{j}-t_{i})w_i(t)\\
&=&L\cdot t_{i}+(t_{j}-t_{i})w_{i}(t).\\
\end{eqnarray*}
\end{proof}

\noindent
\textit{Proof of Theorem~\ref{thm1}.}\
Consider $n=m$ jobs of length $l=(1,\ldots,1,m)$.
Let $L:=2m-1$ denote the total length of the jobs. Define speed vector $t=(m\alpha,\ldots,m\alpha,\alpha)$, where $\alpha$ is a real number that only depends on $m$ and $c$ and will be determined at the end of this section. We will show that if $M$ is deterministic, truthful, envy-free, individually rational and anonymous, it should allocate all jobs to the quickest machine in this instance.

\begin{claim}\label{lemma_lb}
For speed vector $t=(m\alpha,\ldots,m\alpha,\alpha)$ and jobs $l=(1,\ldots,1,m)$, if $M$ is $c$-approximate and $w_{i}(t)\geq1$ for some $i\in\{1,\ldots,m-1\}$, then $$h(t_{-1})\geq (L+\frac{m-1}{Lc})\cdot \alpha.$$
\begin{proof}
Since $M$ is truthful and individually rational, inequality (\ref{equir}) applies, and
$$
h(t_{-1})~\geq~\int_{0}^{\infty}w(x,t_{-1})\dx
~\geq~\int_{0}^{\frac{\alpha}{Lc}}w(x,t_{-1})\dx+\int_{\frac{\alpha}{Lc}}^{\alpha}w(x,t_{-1})\dx+\int_{\alpha}^{m\alpha}w(x,t_{-1})\dx.
$$
Apply $M$ to vector $(x,t_{-1})$. By the local efficiency of $w$, job $m$ should be assigned to the quickest machine. So for $x<\alpha$, $w(x,t_{-1})\geq l_{m}$. When $x<\frac{\alpha}{Lc}$, all the jobs should be assigned to the machine with speed $x$ for a makespan less than $\alpha/c$. Otherwise the makespan is at least $\alpha$, contradicting $M$ is $c$-approximate. Since $w_{i}(m\alpha,t_{-i})\geq1$ for some $i\in\{1,\ldots,m-1\}$, monotonicity implies $w_{i}(x,t_{-1})\geq 1$ for all $x\in(\alpha,m\alpha)$. Since $x\in(\alpha,m\alpha)$ is unique in vector $(x,t_{-1})$, we get $w(x,t_{-1})\geq 1$ by anonymity. Thus
$$
h(t_{-1})~\geq~\int_{0}^{\frac{\alpha}{Lc}}L\dx+\int_{\frac{\alpha}{Lc}}^{\alpha}m\dx+\int_{\alpha}^{m\alpha}1\dx
~=~\frac{1}{c}\alpha+m\alpha-\frac{m}{Lc}\alpha+m\alpha-\alpha
~=~(L+\frac{m-1}{Lc})\alpha.
$$
\end{proof}
\end{claim}

Let $t'=(1,m\alpha,\ldots,m\alpha,\alpha)$. Applying $M$ to $t'$, Lemma~\ref{lemma_envy} implies
$$
h(t'_{-1})-h(t'_{-m})~\leq~L+(\alpha-1)w_{1}(t')~\leq~ L\cdot\alpha.
$$
Since $t'_{-1}=t_{-1}$, this implies
\begin{equation}\label{equup1}
h(t_{-1})\leq L\cdot\alpha + h(t'_{-m}).
\end{equation}
\begin{claim}\label{lemma_ub2}
If $M$ is $c$-approximate, then $h(t_{-1})<L\cdot\alpha+ f(m,c)$, where $f(m,c)=\gamma^{m-1}L+h(\gamma^{m-2},\gamma^{m-1},\ldots, \gamma,1)$ and $\gamma=cL+\epsilon$ for some $0<\epsilon<1$.
\begin{proof}
Define speed vector $t^{(i)}=(\gamma^{i-1},\gamma^{i-2},\ldots,\gamma,1,m\alpha,\ldots,m\alpha)$ of length $m$ for $i\geq 2$, where $\gamma=cL+\epsilon$ for some $0<\epsilon<1$.

Let us consider the allocation $M$ makes to machine 1 for bid vector $t^{(i)}$. The speed of machine 1 is $\gamma^{i-1}\geq \gamma$ for $i\geq 2$. The speed of machine $i$ is 1. The makespan of allocating all jobs to machine $i$ is $L$ while the makespan of allocating at least one job to machine 1 is at least $\gamma=cL+\epsilon$. Since $M$ is $c$-approximate, this means $w_{1}(t^{(i)})=0$. Using Lemma~\ref{lemma_envy}, we have
$$
h(t^{(i)}_{-1})-h(t^{(i)}_{-m})~\leq~{t^{(i)}_{1}}L+(t^{(i)}_{m}-t^{(i)}_{1})w_{1}(t^{(i)})~=~\gamma^{i-1} L.
$$
Since $t^{(i)}_{-m}=t^{(i+1)}_{-1}$, this implies
$
h(t^{(i)}_{-1})-h(t^{(i+1)}_{-1})\leq \gamma^{i-1}L\text{ for } i\in\{2,\ldots, m-1\}.
$
Summing up these inequalities on all $i$, we have
\begin{equation*}
h(t^{(2)}_{-1})-h(t^{(m)}_{-1})\leq L\sum_{i=2}^{m-1}\gamma^{i-1}< \gamma^{m-1}L
\end{equation*}
Since $t'_{-m}=t^{(2)}_{-1}$, we get $h(t'_{-m})<\gamma^{m-1}L+h(t^{(m)}_{-1})=f(m,c)$. Plugging this into (\ref{equup1}) yields
\begin{equation}\label{equup2}
h(t_{-1})<L\cdot\alpha+ f(m,c).
\end{equation}
\end{proof}
\end{claim}

To complete the proof of Theorem~\ref{thm1}, consider speed vector $t$ with $\alpha=\frac{Lc}{m-1}f(m,c)$. If mechanism $M$ does not allocate all jobs to machine $m$, then $w_{i}(t)\geq 1$ for some $i\in\{1,\ldots,m-1\}$. Then Claim~\ref{lemma_lb} implies that $h(t_{-1})\geq \alpha\cdot L+f(m,c)$, contradicting (\ref{equup2}). So $M$ must allocate all jobs to machine $m$ in this case, yielding  a makespan of $(2m-1)\alpha$ while the makespan of the  schedule that assigns job $j$ to machine $j$ for all $j$ is $m\alpha$. Thus, $M$ is $c$-approximate for some $c\geq 2-1/m$. $\hfill{} \qed$

\section{Characterizing Scalable Mechanisms on Two Machines}
We show that known monotone and locally efficient allocations do not have payments to form truthful, envy-free, individually rational, and anonymous mechanisms. (See Section 5 and Appendix B.) 

In this section, we will show that for two machines, there is just one deterministic scalable allocation that can be made truthful, envy-free, individually rational, and anonymous. This allocation turns out to be the same allocation as the VCG mechanism.

\begin{lemma}\label{lemma_int}
Let $w$ be a deterministic and scalable allocation function for 2 machines. For some $k>1$, if $w(x,a)>0$ for all $a>0$ and $x<ka$, then there is some $g(k)>0$ such that
$$\int_{a}^{ka}w(x,a)\dx\geq\int_{\frac{a}{k}}^{a}w(a,x)\dx+g(k)\cdot a.$$
\begin{proof}
For $a<x<ka$, let $x=\frac{a^{2}}{t}$.
\begin{align*}
\int_{a}^{ka}w(x,a)\dx &= \int_{a}^{\frac{a}{k}}w(\frac{a^{2}}{t},a)(-\frac{a^{2}}{t^{2}})\dt\tag{integrate by substitution}\\
&=\int_{\frac{a}{k}}^{a}\frac{a^{2}}{t^{2}}w(a,t)\dt\tag{$w$ is scalable}\\
&=\int_{\frac{a}{k}}^{\frac{k+1}{2k}a}\frac{a^{2}}{t^{2}}w(a,t)\dt+\int_{\frac{k+1}{2k}a}^{a}\frac{a^{2}}{t^{2}}w(a,t)\dt\\
\end{align*}
For $\frac{a}{k}<t<\frac{k+1}{2k}a$ and $k>1$, we have $\frac{a^{2}}{t^{2}}\geq a^{2}/(\frac{k+1}{2k}a)^{2}=\frac{4k^{2}}{(k+1)^{2}}>1$. For $\frac{k+1}{2k}a<t<a$, we have $\frac{a^{2}}{t^{2}}\geq 1$. Therefore,
\begin{eqnarray*}
\int_{a}^{ka}w(x,a)\dx&\geq&\frac{4k^{2}}{(k+1)^{2}}\int_{\frac{a}{k}}^{\frac{k+1}{2k}a}w(a,t)\dt+\int_{\frac{k+1}{2k}a}^{a}w(a,t)\dt \\
&=&\left(\frac{4k^{2}}{(k+1)^{2}}-1\right)\int_{\frac{a}{k}}^{\frac{k+1}{2k}a}w(a,t)\dt+\int_{\frac{a}{k}}^{a}w(a,t)\dt\\
\end{eqnarray*}
We also have
$$\int_{\frac{a}{k}}^{\frac{k+1}{2k}a}w(a,t)\dt~=~\int_{\frac{a}{k}}^{\frac{k+1}{2k}a}w(1,\frac{t}{a})\dt~=~a\int_{\frac{1}{k}}^{\frac{k+1}{2k}}w(1,y)\ dy.$$
The first equality follows the scalability of $w$ and we get the second equality by substituting $t$ with $ay$.
Since $w(x,a)>0$ for all $a>0$ and $x<ka$, we have $w(1,y)>0$ for $y>1/k$.
In sum, $\int_{a}^{ka}w(x,a)\dx\geq\int_{\frac{a}{k}}^{a}w(a,x)\dx+g(k)\cdot a$, where
$g(k)=\left(\frac{4k^{2}}{(k+1)^{2}}-1\right)\int_{\frac{1}{k}}^{\frac{k+1}{2k}}w(1,y)\ dy>0$.
\end{proof}
\end{lemma}
\subparagraph{}
\begin{theorem}\label{thm3}
Let $M=(w,p)$ be deterministic, truthful, envy-free, individually rational, and anonymous. If $w$ is scalable, then for two machines of different speeds, $w$ allocates all jobs to the quickest machine.
\begin{proof}
Let $L$ denote the total length of jobs. First, consider two machines of speed $t_{1}=1$ and $t_{2}=a$ $(a>1)$.
Since $M$ is truthful, envy-free, and anonymous, by Lemma~\ref{lemma_envy}, we have
\begin{equation}\label{equ3}
h(a)-h(1)~\leq~L+(a-1)L~=~ L\cdot a
\end{equation}
Since $w$ is individually rational, $h(1)\geq\int_{0}^{\infty}w(x,1)\geq0$. We will show that $w(ka,a)=0$ for any $k>1$. For a contradiction, assume $w(ra,a)>0$ for some $r>1$. Let $k$ be such that $w(x,a)>0$ for $x<ka$ and $w(x,a)=0$ for $x>ka$. By monotonicity, such a $k$ exists. By the assumption that $w(ra,a)>0$ for some $r>1$, we know that $k>1$. Since $w$ is scalable, we have for any $x>ka$, $w(y,a)=w(a,x)=L$ if $y/a=a/x$, i.e. $y=a^{2}/x<a/k$. Therefore,
\begin{align}
h(a)&~\geq~\int_{0}^{\infty}w(x,a)\dx&\nonumber\\
&~=~\int_{0}^{\frac{a}{k}}L\dx+\int_{\frac{a}{k}}^{a}w(x,a)\dx+\int_{a}^{ka}w(x,a)\dx&\nonumber\\
&~\geq~\frac{L}{k}a+\int_{\frac{a}{k}}^{a}w(x,a)\dx+\int_{\frac{a}{k}}^{a}w(a,x)\dx+g(k)\cdot a&\tag{Lemma~\ref{lemma_int}}\nonumber\\
&~\geq~\frac{L}{k}a+\int_{\frac{a}{k}}^{a}\left(w(x,a)+w(a,x)\right)\dx+g(k)\cdot a&\nonumber\\
&~=~\frac{L}{k}a+\int_{\frac{a}{k}}^{a}L\dx+g(k)\cdot a&\nonumber\\
&~=~(L+g(k))a&\label{equstar}
\end{align}
Take $a>h(1)/g(k)$. We have $h(a)>aL+h(1)$ from (\ref{equstar}). This contradicts~(\ref{equ3}).
\end{proof}
\end{theorem}

\section{Payments for Known Allocation Rules}
Although the VCG mechanism is truthful, envy-free, individually rational, and anonymous, it does not have a good approximation guarantee for makespan. In \cite{kovacs}, the \lpt algorithm is described and shown to be monotone decreasing. In this section, we will show that \lpt is locally efficient and no payment function can form an envy-free, truthful, individually rational, and anonymous mechanism with the \lpt algorithm. We also prove a similar result for randomized mechanisms in Appendix B: the randomized 2-approximation algorithm in~\cite{Archer,at}, whose expected allocation is monotone decreasing and locally efficient,  admits no payment function that can make it simultaneously truthful-in-expectation, envy-free-in-expectation, individually rational, and anonymous. We can show similar results with similar proofs for the allocations in~\cite{aas, ck10, dddr}. 

The \lpt algorithm is the following: Let $w^j_i$ be the workload of machine $i$ before job $j$ is assigned. Assume the jobs are indexed so that $l_1\geq l_2\geq\ldots\geq\l_m$. Note that this algorithm rounds the speeds and hence is not scalable.\par
 \begin{algorithm}
 \caption{\lpt Algorithm}
 \begin{algorithmic}[1]
 \STATE Define rounded speed of machine $i$ to be $s_i:=2^{\lceil\log b_i\rceil}$.
 \FOR {$j=1$ to $m$}
 \STATE Assign job $j$ to machine $i$ that minimizes $(w^j_i+l_j)\cdot s_i$.
 \ENDFOR
 \STATE Among machines of same rounded speed, reorder bundles on these machines so that a machine with smaller bid gets more jobs.
 \end{algorithmic}
 \end{algorithm}
 \begin{lemma}\label{lemma41}
\lpt is locally efficient.
\begin{proof}
We need to show that $w_{i}(b)\leq w_{k}(b)$ for any $b_{i}> b_{k}$. If $s_{i}=s_{k}$, then step 5 ensures $w_{i}(b)\leq w_{k}(b)$. So suppose $s_{i}>s_{k}$.
Consider the last job, $j$, assigned to machine $i$. Since job $j$ is assigned to machine $i$ rather than machine $k$, it should be that
$$(w^{j}_{i}+l_{j})s_{i}\leq (w^{j}_{k}+l_{j})s_{k},$$
where $w^{j}_{i}$ is the workload of machine $i$ before job j is assigned.
Thus,
$$w^{j}_{k}+l_{j}\geq \frac{s_{i}}{s_{k}}(w^{j}_{i}+l_{j})\geq 2(w^{j}_{i}+l_{j}).$$
That is
$w^{j}_{k}\geq 2w^{j}_{i}+l_{j}.$
Since $w_{k}(b)\geq w^{j}_{k}$ and $w_{i}(b)=w^{j}_{i}+l_{j}$, we get $w_{k}(b)\geq w_{i}(b)$.
\end{proof}
\end{lemma}
\begin{theorem}\label{thm2}
There is no payment function that will make $\lpt$simultaneously truthful, envy-free, individually rational, and anonymous.
\begin{proof}
Let $w$ denote the allocation of the \lpt algorithm. For a contradiction, assume there exists a payment function $p$ such that mechanism $M=(w,p)$ is truthful, envy-free, individually rational, and anonymous.\\
Apply $M$ to the problem of two jobs with lengths $l_{1}=2$ and $l_{2}=1$, and two machines with speeds $t_{1}=1,t_{2}=a$ where $a>1$ and $a$ is a power of 2. By Lemma~\ref{lemma_envy}, we have
\begin{equation}\label{thm2eq1}
h(a)-h(1)\leq 3+(a-1)\cdot 3=3a.
\end{equation}
Since $M$ is individually rational, we also have
\begin{eqnarray*}
h(a)&\geq&\int_{0}^{\infty}w(x,a)\dx\\
&\geq&\int_{0}^{\frac{a}{4}}w(x,a)\dx+\int_{\frac{a}{4}}^{a}w(x,a)\dx+\int_{a}^{2a}w(x,a)\dx\\
&\geq&\frac{a}{4}w(\frac{a}{4},a)+\int_{\frac{a}{4}}^{a}w(x,a)\dx+a\cdot w(2a,a),\\
\end{eqnarray*}
where the last inequality follows the monotonicity of $w$. Since $a$ is a power of 2, the $\lpt$ algorithm ensures $w(\frac{a}{4},a)=3$ and $w(2a,a)=1$. Since $w$ is locally efficient, for any $\frac{a}{4}<x<a$, a machine with speed $x$ gets at least job one, i.e., $w(x,a)\geq2$. Therefore, $h(a)\geq \frac{a}{4}\cdot 3+\frac{3a}{4}\cdot 2+a\cdot 1=\frac{13}{4}a$.
Now take $a=8h(1)$, we have $h(a)\geq \frac{13}{4}a=26h(1)$ and $h(a)\leq h(1)+3a=25h(1)$ from~(\ref{thm2eq1}), a contradiction.
\end{proof}
\end{theorem}

Theorem~\ref{thm2} implies that local efficiency and monotonicity of an allocation are not sufficient for the existence of a payment function to form an envy-free, truthful, individually rational, and anonymous mechanism. This insufficiency of monotonicity and local efficiency still holds, even if the allocation function is also scalable. See Proposition~\ref{prop12} in Appendix~\ref{appendixC} for more details.

\section{Open Questions}
In this paper, we establish an approximation lower bound $2-1/m$ for any deterministic, envy-free, truthful, individually rational, and anonymous mechanism while the upper bound is $m$ given by the VCG mechanism. So one open question is whether the VCG mechanism is the best among all truthful and envy-free mechanisms.\par
The proof of Lemma~\ref{lemma_envy} implicitly gives a characterization of
\emph{mechanisms} that are truthful, envy-free, and anonymous.
Ideally, there would be a characterization of \emph{allocations}
for which there exist prices that make the resulting mechanism truthful
and envy-free.  Another interesting question is whether there is a
characterization of truthful and envy-free mechanisms for $\qcmax$.

\bibliographystyle{plain} 
\bibliography{AMD} 

\appendix
\noindent{\bf\Large Appendix}
\section{Randomized mechanisms}
In this section, we define the randomized mechanisms for $\qcmax$ in our setting.
A randomized mechanism $(w,p)$ is \emph{truthful-in-expectation}, if each agent $i$ maximizes his expected utility by bidding his true value $t_{i}$, i.e., for all agent $i$, all possible $t_{i},b_{i}$ and $b_{-i}$, $$\ep{p_{i}(t_{i},b_{-i})-t_{i}w_{i}({t_{i},b_{-i}})}\geq \ep{p_{i}(b_{i},b_{-i})-t_{i}w_{i}({b_{i},b_{-i}})}.$$
A randomized mechanism $(w,p)$ is \emph{envy-free-in-expectation} if no agent wishes to switch his expected allocation and expected payment with another, i.e., for all $i,j\in\{1,\ldots, m\}$ and all bids $b$, $$\ep{p_{i}(b)}-b_{i}\ep{w_{i}(b)}\geq \ep{p_{j}(b)}-b_{i}\ep{w_{j}(b)}.$$
A randomized mechanism $(w,p)$ is anonymous if for every bid vector $b=(b_{1},\ldots,b_{m})$, every $k$ such that $b_{k}$ is unique and every $l\neq k$,
$$\ep{w_{l}(\ldots,b_{k-1},b_{l},b_{k+1},\ldots,b_{l-1},b_{k},b_{l+1},\ldots)}=\ep{w_{k}(b)}$$ and
$$\ep{p_{l}(\ldots,b_{k-1},b_{l},b_{k+1},\ldots,b_{l-1},b_{k},b_{l+1},\ldots)}=\ep{p_{k}(b)}.$$

So we seek randomized mechanisms that are truthful-in-expectation, envy-free-in-expectation, individually rational, and anonymous.
For a randomized mechanism that is truthful-in-expectation, Lemma 1 still holds after replacing all $w_{i}(b)$ with $\ep{w_{i}(b)}$~\cite{at}. By Lemma 1, the payments are deterministic of a randomized mechanism that is truthful-in-expectation. It is easy to check that~(\ref{equir}), Lemma~\ref{lemma_1} and Lemma~\ref{lemma_envy} still hold after replacing the allocation of a machine with its expected allocation for all machines.

\section{}
In \cite{at}, a randomized 3-approximation algorithm (see Algorithm 2) is described and its expected allocation is shown to be monotone decreasing. In this section, we will show that its expected allocation is locally efficient and there is no payment that makes it a truthful-in-expectation, envy-free-in-expectation, individually rational and anonymous mechanism.

 \begin{algorithm}
 \caption{Randomized 3-approximation Algorithm by Archer and Tardos~\cite{at}}
 \begin{algorithmic}[1]
 \REQUIRE Jobs and machines are indexed so that $l_1\geq l_2\geq\ldots\geq\l_n$ and $b_{1}\leq b_{2},\ldots,\leq b_{m}$.
 \STATE Compute a lower bound of makespan $$T_{LB}(b)=\max_{j}\min_{i}\max
 \left \{b_{i}l_{j},\frac{\sum_{k=1}^{j}l_{k}}{\sum_{r=1}^{i}\frac{1}{b_{r}}}\right\}$$
  \STATE For each machine $i$, create a bin $i$ of size $s_{i}(b)=T_{LB}(b)/b_{i}$.
 \STATE Assign jobs $1,2,\ldots,(k-1)$ to bin 1, where  $k$ is the first job that would cause the bin to overflow. Then assign to bin 1 a piece of job $k$ exactly as large as the remaining capacity of bin 1. Continue by assigning jobs to bin 2, starting with the rest of job $k$ and so on, until all jobs are assigned.
 \STATE For each job $j$, assign job $j$ to machine $i$ with probability equal to the proportion of job $j$ that is fractionally assigned to bin $i$.
 \end{algorithmic}
 \end{algorithm}

\begin{theorem}
There is no payment function that will make Algorithm 2 simultaneously truthful-in-expectation, envy-free-in-expectation, individually rational, and anonymous.
\begin{proof}
Let $w$ denote the allocation of Algorithm 2. For a contradiction, assume there exists a payment function $p$ such that mechanism $M=(w,p)$ is truthful-in-expectation, envy-free-in-expectation, individually rational and anonymous.\\

Apply $M$ to the problem of two jobs with lengths $l_{1}=2$ and $l_{2}=1$, and two machines with speeds $t_{1}=1,t_{2}=a$ where $a>1$. By Lemma~\ref{lemma_envy}, we have
\begin{equation}\label{equat1}
h(a)-h(1)\leq 3+(a-1)\cdot 3=3a,
\end{equation}
where $h(1)\geq \int_{0}^{\infty}w(x,1)\dx\geq 0$. The  expected workload of machine $i$ is equal to the size of bin $i$ by step 4, i.e., $\ep{w_{i}(b)}=s_{i}(b)$. Algorithm 2 also ensures that for all $c>0$,
$$s_{i}(c\cdot b)=\frac{T_{LB}(c\cdot b)}{c\cdot b_{i}}=\frac{c\cdot T_{LB}(b)}{c\cdot b_{i}}=s_{i}(b).$$
So we get $\ep{w_{i}(c\cdot b)}=\ep{w_{i}(b)}$ for all $c>0$. This implies, for two machines,
$$\int_{0}^{\infty}\ep{w(x,a)}\dx~=~\int_{0}^{\infty}\ep{w(\frac{x}{a},1)}\dx~=~a\int_{0}^{\infty}\ep{w(t,1)}\dt,$$
where the last equality is obtained by substituting $x$ with $at$. Since $M$ is individually rational, we have
\begin{eqnarray*}
h(a)~\geq~\int_{0}^{\infty}\ep{w(x,a)}\dx~=~a\int_{0}^{\infty}\ep{w(x,1)}\dx~=~(3.5+\ln3-\ln2)a,
\end{eqnarray*}
where we get the last equation by computing  $\int_{0}^{\infty}\ep{w(x,1)}\dx$ using the allocation of Algorithm 2. Now take $a=2h(1)$, we have $h(a)\geq (3.5+\ln3-\ln2)a>7h(1)$, contradicting~(\ref{equat1}).
\end{proof}
\end{theorem}

Since the expected workload of machine $i$ equals to $s_{i}(b)=T_{LB}(b)/b_{i}$, we get $\ep{w_{i}(b)}=T_{LB}(b)/b_{i}\geq T_{LB}(b)/b_{j}=\ep{w_{j}(b)}$ for $b_{i}\leq b_{j}$. So the expected allocation of Algorithm 2 is locally efficient.

\section{}\label{appendixC}
\begin{prop}\label{prop12}
Local efficiency, monotonicity and scalability of an allocation are not sufficient for the existence of a payment function to form an envy-free, truthful, individually rational, and anonymous mechanism.
\begin{proof}
Define allocation $w$ for 2 machines to be the allocation that minimizes the makespan. If there are more than one such allocations, let $w$ be the one that also minimizes the total completion time.
It is easy to verify that $w$ is unique. Hence $w$ is well-defined and anonymous. $w$ is also locally efficient and scalable, since it minimizes the makespan. Now, we show that $w$ is monotone.
Let the allocation of $w$ for 2 machines with bids $(b_1,b_2)$ be $\mathcal{O}=(L_1,L_2)$. Assume w.l.o.g that $b_1L_1\geq b_2L_2$. If machine 1 increases its bid, its allocation can only go down.
Consider the allocation of $w$: $\mathcal{O'}=(L'_1,L'_2)$ for 2 machines with bids $(b_1,b'_2)$, where $b'_2>b_2$. For a contradiction, assume $L'_2>L_2$. So $L'_1<L_1$. If the makespan of $\mathcal{O'}$ is
$b_1L'_1$, i.e. $b_1L'_1\geq b'_2L'_2$, then we have $b_2L'_2<b'_2L'_2\leq b_1L'_1<b_1L_1$. So $\max\{b_2L'_2,b_1L'_1\}<b_1L_1$. It contradicts the optimality of $\mathcal{O}$.
If the makespan of $\mathcal{O'}$ is $b'_2L'_2$, i.e. $b'_2L'_2\geq b_1L'_1$, then $b'_2L'_2\leq \max\{b_1L_1,b'_2L_2\}$ by the optimality of $\mathcal{O'}$.
Since $b'_2L'_2>b'_2L_2$, we have $b'_2L'_2\leq b_1L_1$. So $b_2L'_2<b_1L_1$. Since $b_1L'_1<b_1L_1$, we have $\max\{b_2L'_2,b_1L'_1\}<b_1L_1$. It contradicts the optimality of $\mathcal{O}$.
Therefore, $L'_2\leq L_2$ and $w$ is monotone.
Since $w$ is different from the VCG allocation for 2 machines, this theorem follows from Theorem~\ref{thm3}.
\end{proof}
\end{prop}

\end{document}